\documentclass[11pt]{article}
\usepackage{times}
\usepackage{amsmath}
\usepackage{amsfonts}
\usepackage{latexsym}
\usepackage{tikz,pgflibrarysnakes,pgflibraryplotmarks}
\usepackage{graphicx}
\usepackage{url,hyperref}

\advance \topmargin by -1.0in
\advance \oddsidemargin by -0.8in
\newcommand{\myparskip}{3pt}
\parskip \myparskip

\newtheorem{lemma}{Lemma}[section]

\newtheorem{definition}[lemma]{Definition}

\newtheorem{prop}[lemma]{Proposition}

\newtheorem{conjecture}{Conjecture}

\newenvironment{proof}{\vspace{-0.15in}\noindent{\bf Proof:}}%
        {\hspace*{\fill}$\Box$\par}
        {\hspace*{\fill}$\Box$\par\vspace{4mm}}
        {\hspace*{\fill}$\Box$\par}

\newcommand{\eps}{\epsilon}

\newcommand{\ceil}[1]{\lceil #1 \rceil}
\newcommand{\floor}[1]{\lfloor #1 \rfloor}
\newcommand{\opt}{\text{\sc OPT}}

\newcommand{\deadline}{{\sc Orient-Deadline}}
\newcommand{\tw}{{\sc Orient-TW}}

\DeclareMathAlphabet{\mathpzc}{OT1}{pzc}{m}{it}

\newcommand{\lmax}{L_{\max}}
\newcommand{\lmin}{L_{\min}}

\newcommand{\cR}{{\cal R}}

\begin{document}

\title{Approximation Algorithms for Orienteering with Time Windows}
\author{
Chandra Chekuri\thanks{Dept.\ of Computer Science, University of Illinois, Urbana, IL 61801.
Partially supported by NSF grant CCF 07-28782. {\tt chekuri@cs.uiuc.edu}}
\and
Nitish Korula\thanks{Dept.\ of Computer Science, University of Illinois, Urbana, IL 61801.
Partially supported by NSF grant CCF 07-28782. {\tt nkorula2@uiuc.edu}}
}

\maketitle

\begin{abstract}
  Orienteering is the following optimization problem: given an
  edge-weighted graph (directed or undirected), two nodes $s,t$ and a
  time limit $T$, find an $s$-$t$ {\em walk} of total length at most
  $T$ that maximizes the number of {\em distinct} nodes visited by the
  walk. One obtains a generalization, namely orienteering with {\em
  time-windows} (also referred to as TSP with time-windows), if each
  node $v$ has a specified time-window $[R(v), D(v)]$ and a node $v$
  is counted as visited by the walk only if $v$ is visited during its
  time-window. For the time-window problem, an $O(\log \opt)$
  approximation can be achieved even for directed graphs if the
  algorithm is allowed {\em quasi-polynomial} time.  However, the best
  known polynomial time approximation ratios are $O(\log^2 \opt)$ for
  undirected graphs and $O(\log^4 \opt)$ in directed graphs. In
  this paper we make some progress towards closing this discrepancy,
  and in the process obtain improved approximation ratios in several
  natural settings.

  Let $L(v) = D(v) - R(v)$ denote the length of the time-window for
  $v$ and let $\lmax = \max_v L(v)$ and $\lmin = \min_v L(v)$. Our
  results are given below with $\alpha$ denoting the known
  approximation ratio for orienteering (without time-windows). Currently
  $\alpha = (2+\eps)$ for undirected graphs and $\alpha = O(\log^2 \opt)$
  in directed graphs.
  \begin{itemize}
  \item An $O(\alpha \log \lmax)$ approximation when $R(v)$ and $D(v)$
   are integer valued for each $v$.
  \item An $O(\alpha \max\{\log \opt, \log \frac{L_{max}}{L_{min}}\})$
    approximation.
  \item An $O(\alpha \log \frac{L_{max}}{L_{min}})$ approximation when
    no start and end points are specified.
  \end{itemize}
  In particular, if $\frac{L_{max}}{L_{min}}$ is poly-bounded, we
  obtain an $O(\log n)$ approximation for the time-window problem in
  undirected graphs.
\end{abstract}

\section{Introduction}
\label{sec:intro}

In the orienteering problem, we are given an edge-weighted graph
$G(V,E)$, two vertices $s,t \in V$, and a time limit $T$. The goal is
to find a walk that begins at $s$ at time $0$, reaches $t$ before time
$T$, and visits as many vertices as possible. (The weight or length of
an edge denotes the time taken to travel from one endpoint to the
other.)  Note that a vertex may be visited many times, but is only
counted once in the objective function. In other words, the goal is to
find an $s$-$t$ walk of total length at most $T$ that maximizes the
number of distinct vertices visited by the walk. This problem is also
referred to as the {\em point-to-point} orienteering problem to
distinguish it from two special cases: only the start vertex $s$ is
specified, or neither $s$ nor $t$ are specified. Here we consider a
more general problem, namely orienteering with time-windows. In this
problem, we are additionally given a time-window (or interval)
$[R(v), D(v)]$ for each vertex $v$. A vertex is counted as visited
only if the walk visits $v$ at some time $t \in [R(v), D(v)]$. For
ease of notation, we use \tw~to refer to the problem of orienteering
with time-windows. A problem of intermediate complexity is the one
in which $R(v) = 0$ for all $v$. We refer to this problem as orienteering
with deadlines (\deadline); it has also been called the Deadline-TSP
problem in \cite{timewindow}. The problem where vertices have release
times but not deadlines (that is, $D(v) = \infty$ for all $v$) is
entirely equivalent to \deadline.

The orienteering problem and the time-window generalization are
intuitively appealing variants of TSP. They also arise in practical
applications of vehicle routing and scheduling \cite{vehicle_book}.
Even in undirected graphs these problems are NP-Hard and also APX-hard
\cite{orienteering}. In fact \tw~is NP-hard even on the line
\cite{Tsitsiklis92}. Although these problems are natural
and simple to state, the first non-trivial approximation algorithm for
undirected graphs appeared only a few years ago \cite{orienteering},
and the first polynomial time approximation algorithm for directed
graphs appeared only this year \cite{NagarajanR07, CKP};
earlier, a constant factor approximation was known for points in the
Euclidean plane \cite{ArkinMN98}.  Table~\ref{table:results}
summarizes the best known approximation ratios.

\begin{table}
\begin{center}
\begin{tabular}{|c|c|c|c|}
\hline
 & Points in Euclidean Space & Undirected Graphs & Directed Graphs \\ \hline
Orienteering & $(1+\eps)$ \cite{ChenH06} & $(2+\eps)$ \cite{CKP} & $O(\log^2 \opt)$
             \cite{NagarajanR07,CKP} \\ \hline
\deadline & $O(\log \opt)$ & $O(\log \opt)$ \cite{timewindow} & $O(\log^3 \opt)$ \\ \hline
\tw~     & $O(\log^2 \opt)$ & $O(\log^2 \opt)$ \cite{timewindow} & $O(\log^4 \opt)$ \\ \hline
\end{tabular}
\end{center}
\caption{Known approximation ratios for orienteering and orienteering with time-windows. Entries without a citation come from combining results for orienteering
with the results from \cite{timewindow} in a black-box fashion. \label{table:results}}
\end{table}

In \cite{ChekuriP05}, a recursive greedy algorithm is given for the
orienteering problem when the reward function is a given monotone
submodular set function\footnote{A function $f: 2^V \rightarrow \cR^+$
  is a montone submodular set function if $f$ satisfies the following
  properties: (i) $f(\emptyset) = 0$, $f(A) \ge f(B)$ for all $A
  \subseteq B$ and (ii) $f(A) + f(B) \ge f(A \cup B) + f(A \cap B)$
  for all $A, B \subseteq V$.}  $f$ on $V$, and the objective is to
maximize $f(S)$ where $S$ is the set of vertices visited by the walk.
Several non-trivial problems, including \tw, can be captured by using
different submodular functions. The algorithm
from \cite{ChekuriP05} provides an $O(\log \opt)$ approximation in
directed graphs, but it runs in {\em quasi-polynomial} time. Thus, we
make the following natural conjecture:

\begin{conjecture}
There is a polynomial time $O(\log \opt)$ approximation
for orienteering with time-windows in directed (and undirected) graphs.
\end{conjecture}

As can be seen from Table~\ref{table:results}, even in undirected
graphs the current best ratio is $O(\log^2 \opt)$. Our primary
motivation is to close the gap between the ratios achievable in
polynomial and quasi-polynomial time respectively. We remark that the
quasi-polynomial time algorithm in \cite{ChekuriP05} is quite
different from all the other polynomial time algorithms, and it does
not appear easy to find a polynomial time equivalent. In this paper we
make some progress in closing the gap, while also obtaining some new
insights. An important aspect of our approach is to understand the
complexity of the problem in terms of the maximum and minimum time-window
lengths. Let $L(v) = D(v) - R(v)$ be the length of the time-window of $v$.
Let $\lmax = \max_v L(v)$ and $\lmin = \min_v L(v)$. Our results depend on
the ratio $L = \lmax/\lmin$.\footnote{If $\lmin = 0$, consider the set of
vertices which have zero-length time-windows. If this includes a
significant fraction of the vertices of an optimal solution, use
dynamic programming to get a $O(1)$-approximation. Otherwise,
we can ignore these vertices and assume $\lmin > 0$ without
losing a significant fraction of the optimal reward.}
We define this parameter following the work of Frederickson
and Wittman \cite{FW07}; they showed that a constant factor
approximation is achievable in undirected graphs if all time-windows
are of the same length (that is, $L=1$) and the end points of the walk are
not specified. We believe this is a natural parameter to consider in the
context of time-windows. In many practical settings $L$ is likely to be
small, and hence, algorithms whose performance depends on $L$ may be
better than those that depend on other parameters. In \cite{timewindow}
an $O(\log D_{\max})$ approximation is given for \tw~in undirected graphs
where $D_{\max} = \max_v D(v)$ and only the start vertex $s$ is specified
(here it is assumed that all the input is integer valued). We believe
that $\lmax$ is a better measure than $D_{\max}$ for \tw; note that
$\lmax \le D_{\max}$ for all instances.

\medskip
\noindent {\bf Results:} We obtain results for both undirected and
directed graphs.  Our results are for \tw~and use
an algorithm for the point-to-point orienteering problem as a black
box. Letting $\alpha$ denote the approximation ratio for the
orienteering problem, we have the following results. 
\begin{itemize}
\item An $O(\alpha \log \lmax)$ approximation when $R(v)$ and $D(v)$
  are integer valued for each $v$.
\item An $O(\alpha \max\{\log \opt, \log L\})$
  approximation.
\item An $O(\alpha \log L)$ approximation when
  no start and end points are specified.
\end{itemize}

We briefly compare our results to previous results to put the
improvements in a proper context. We focus on undirected graphs where
$\alpha = O(1)$.  The $O(\log \lmax)$ approximation improves on the
$O(\log D_{\max})$ approximation from \cite{timewindow} in several ways.
First, $\lmax \le D_{\max}$ in all instances and is considerably smaller in
many instances. Second, our algorithm applies to directed graphs while
the algorithm in \cite{timewindow} is applicable only for undirected
graphs. Third, our algorithm is for the point-to-point version while
the one in \cite{timewindow} does not guarantee that the walk ends at $t$.
The $O(\max\{\log \opt, \log L\})$ ratio improves the $O(\log^2 \opt)$
approximation in \cite{timewindow} on instances where $L$ is not too large;
in particular if $L$ is poly-bounded then the ratio we obtain is $O(\log
n)$ while the previous guarantee is $O(\log^2 n)$ when expressed as a
function of $n$, the number of vertices in $G$. Finally, our bound
of $O(\log L)$ strictly generalizes the result of \cite{FW07} who
consider only the case of $L = 1$.

Our results are obtained using relatively simple ideas.  Nevertheless,
we believe that they are interesting, useful and shed more light on
the complexity of the problem. In particular we are optimistic that
some of these ideas will lead to an $O(\log n)$ approximation for the
time-window problem in undirected graphs even when $L$ is not
poly-bounded.

\medskip {\bf Related Work:} Table~\ref{table:results} refers to a
good portion of the recent work on approximation algorithms for
orienteering and related problems.  The first non-trivial
approximation algorithm for orienteering was a $(2+\eps)$-approximation
in the Euclidean plane \cite{ArkinMN98}.  In \cite{orienteering}, the
authors showed how one can use an approximation for the $k$-stroll problem
(here, the goal is to find a minimium length $s$-$t$ walk that visits $k$
vertices) to obtain an approximation for orienteering. In undirected
graphs an approximation for $k$-stroll can be obtained using an
approximation for the more well-studied $k$-MST problem although one
can obtain improved ratios for $k$-stroll using related ideas
\cite{ChaudhuriGRT03}. In \cite{timewindow}, orienteering is used as a
black box for \deadline~and \tw. For directed graphs a bi-criteria
approximation for $k$-stroll was only recently obtained
\cite{NagarajanR07,CKP} and this led to the first
poly-logarithmic approximation for orienteering and \tw.
The recursive greedy algorithm from \cite{ChekuriP05}, as
discussed before, is based on a different approach. Other special
cases have been considered in the literature. For points on a line, an
$O(\min\{\log n, \log L\})$ approximation is given in
\cite{Bar-YehudaES}. When the number of distinct time-windows is a fixed
constant, \cite{ChekuriK04} gives an $O(\alpha)$ approximation; here
$\alpha$ is the approximation ratio for orienteering.
As we already mentioned, \cite{FW07} considered the case of equal
length time-windows.

\section{Preliminaries and General Techniques}
\label{sec:prelims}
Much of the prior work on orienteering with time-windows, following
\cite{timewindow}, has used the same general technique or can be cast in this
framework: Use combinatorial methods to reduce the problem to a collection of
sub-problems where the time-windows can be ignored. Each sub-problem has a
subset of vertices $V'$, start and end vertices $s',t' \in V'$, and a
time-interval $I$ in which we must travel from $s'$ to $t'$, visiting as many
vertices of $V'$ within their time windows as possible. However, the
sub-problem is constructed such that the time-window for every vertex in $V'$
entirely contains the interval $I$. Therefore, the sub-problem is really an
instance of the orienteering problem (without time-windows). An approximation
algorithm for orienteering can be used to solve each sub-problem, and these
solutions can be pasted together using dynamic programming. The next two
sub-sections describe this framework.

One consequence of using this general method is that the techniques we
develop apply to both directed and undirected graphs; while
solving a sub-problem we use either the algorithm for orienteering
on directed graphs, or the algorithm for undirected graphs.
Better algorithms for either of these problems would immediately
translate into better algorithms for orienteering with time-windows.

Subsequently, we mainly study undirected graphs, and state our
results in that context. The corresponding approximation ratios
for directed graphs are a factor of $O(\log^2 \opt)$ higher; 
this is simply because $O(\log^2 \opt)$ is the ratio between the
current best approximations for orienteering (without time-windows)
in directed and undirected graphs.

Recall that $\lmax$ and $\lmin$ are the lengths of the longest and
shortest time time-windows respectively, and $L$ is the ratio
$\frac{\lmax}{\lmin}$.
We first provide two algorithms with the following guarantees:
\begin{itemize}
\item $O(\log \lmax)$, if the release time and deadline of every
vertex are integers.
\item $O(\log n)$, if $L \le 2$.
\end{itemize}

The second algorithm immediately leads to a $O(\log n \times \log L
)$-approximation for the general time-window problem, which is already an
improvement on $O(\log^2 n)$ when the ratio $L$ is small. However, we
can combine the first and second algorithms to obtain a $\max\{O(\log n),
O(\log L)\}$-approximation for orienteering with time-windows.

Throughout this paper, we use $R(v)$ and $D(v)$ to denote
(respectively) the release time and deadline of a vertex $v$.
We also use the word \emph{interval} to denote a time
window; $I(v)$ denotes the interval $[R(v), D(v)]$. Typically, we
use `time-window' when we are interested in the start and end
points of a window, and `interval' when we think of a window as
an interval along the `time axis'. 

For any instance $X$ of \tw, we let $\opt(X)$
denote the reward collected by an optimal solution for $X$. When
the instance is clear from context, we use $\opt$ to denote this
optimal reward.

\subsection{The General Framework}

As described at the beginning of section \ref{sec:prelims}, our
general method to solve \tw~is to reduce
the problem to a set of sub-problems without time-windows.
Given an instance of \tw~on a graph $G(V,E)$,
suppose $V_1, V_2, \ldots V_m$ partition $V$, and we can
associate times $R_i$ and $D_i$ with each $V_i$ such that each
of the following conditions holds:
\vspace{-0.1in}
\begin{itemize}
  \item For each $v \in V_i$, $R(v) \le R_i$ and $D(v) \ge D_i$.
  \item For $1 \le i < m$, $D_i < R_{i+1}$.
  \item An optimal solution visits any vertex in $V_i$ during
  $[R_i, D_i]$.
\end{itemize}
\vspace{-0.1in}
Then, we can solve an instance of the orienteering problem in each
$V_i$ separately, and combine the solutions using dynamic
programming. The approximation ratio for such ``composite'' solutions
would be the same as the approximation ratio for the orienteering
problem. We refer to an instance of \tw~in which we can construct
such a partition of the vertex set (and solve the sub-problems
separately) as a \emph{modular instance}. Subsection
\ref{subsec:dyn-prog} describes a dynamic program that can solve
modular instances.

Unfortunately, given an arbitrary instance of \tw, it is unlikely
to be a modular instance. Therefore, we define restricted versions
of a given instance:

\begin{definition} Let $A$ and $B$ be instances of the time-window
problem on the same underlying graph (with the same edge-weights),
and let $I_A(v)$ and $I_B(v)$ denote the intervals for vertex $v$ in
instances $A$ and $B$ respectively. We say that $B$ is a
\emph{restricted version} of $A$ if, for every vertex $v$, $I_B(v)$
is a sub-interval of $I_A(v)$.
\end{definition}

Clearly, a walk that gathers a certain reward in a restricted version
of an instance will gather at least that reward in the original instance.
We attempt to solve \tw~by constructing a set of
restricted versions that are easier to work with. Typically, the
construction is such that the reward of an optimal solution in at
least one of the restricted versions is a significant fraction of the
reward of an optimal solution in the original instance. Hence, an
approximation to the optimal solution in the `best' restricted version
leads us to an approximation for the original instance.

This idea leads us to the next proposition, the proof of which is
straightforward, and hence omitted.
\begin{prop}\label{prop:restricted-versions}
Let $A$ be an instance of \tw~on a graph $G(V,E)$.
If $B_1, B_2, \ldots B_\beta$ are restricted versions of $A$, and for
all vertices $v \in V$, $I_A(v) = \bigcup_{1 \le i \le \beta}I_{B_i}(v)$,
there is some $B_j$ such that $\opt(B_j) \ge \frac{\opt(A)}{\beta}$.
\end{prop}

The restricted versions we construct will usually be modular instances
of \tw. Therefore, the general algorithm for \tw~is:
\begin{enumerate}
\item Construct a set of $\beta$ restricted versions of the
given instance; each restricted version is a modular instance.
\item Pick the best restricted version (enumerate over all choices),
find an appropriate partition, and use an $\alpha$-approximation for
orienteering together with dynamic programming to solve that instance.
\end{enumerate}
It follows from the previous discussion that this gives a 
$(\alpha \times \beta)$-approximation for \tw.
We next describe how to solve modular instances of \tw.

\subsection{A dynamic program for modular instances}
\label{subsec:dyn-prog}

Recall that a modular instance is an instance of \tw~on a graph
$G(V,E)$ in which the vertex set $V$ can be
partitioned into $V_1, V_2, \ldots V_m$, such that an optimal
solution visits vertices of $V_i$ after time $R_i$ and before $D_i$.
For any vertex $v \in V_i$, $R(v) \le R_i$ and $D(V) \ge D_i$.
Further, vertices of $V_i$ are visited before vertices of $V_j$,
for all $j > i$. 

To solve a modular instance, for each $V_i$ we could `guess' the
first and last vertex visited by an optimal solution, and guess
the times at which this solution visits the first and last vertex.
If $\alpha$ is the approximation ratio of an algorithm for
orienteering, we find a path in each $V_i$ that collects an
$\alpha$-fraction of the optimal reward, and combine these solutions.

More formally, we use the following dynamic program: For any
$u,v \in V_i$, consider the graph induced by $V_i$, and let
$\opt(u,v,t)$ denote the optimal reward collected by any walk from
$u$ to $v$ of length at most $t$ (ignoring time-windows).
Now, define $\Pi_i(v,T)$ for $v \in V_i, R_i \le T \le D_i$ as the
optimal reward collected by any walk in $G$ that begins at $s$ at time
0, and ends at $v$ at time $T$. Given $\opt(u,v,t)$, the following
recurrence allows us to easily compute $\Pi_i(v,T)$:
\[ \Pi_i(v,T) = \max_{u \in V_i, w \in V_{i-1}, t \le T -R_i} \opt(u,v,t) + 
\Pi_{i-1} (w, T-t-d(w,u)). \]

Of course, we cannot exactly compute $\opt(u,v,t)$; instead, we use an
$\alpha$-approximation algorithm for orienteering to compute
an approximation to $\opt(u,v,t)$ for all $u,v \in V_i$, $t \le D_i - R_i$.
This gives an $\alpha$-approximation to $\Pi_i(v,T)$ using the recurrence
above.

Unfortunately, the running time of this algorithm is polynomial in
$T$; this leads to a pseudo-polynomial algorithm. To obtain a polynomial-time
algorithm, we use a standard technique of dynamic programming based on reward
instead of time (see \cite{orienteering}, \cite{ChekuriP05}).  Using standard
scaling tricks for maximization problems, one can reduce the problem
with arbitrary rewards on the vertices to the problem where the reward
on each vertex is $1$; the resulting loss in approximation can be made
$(1 + o(1))$.

To construct a dynamic program based on reward instead of time, we ``guess''
the reward $k_i$ collected by an optimal solution in each $V_i$. However,
it is important that we do not try to find an (approximately) shortest
path in each $V_i$ that collects reward $k_i$, since taking slightly too
much time early on can have bad consequences for later $V_i$s. Instead,
we use binary search to compute the shortest walk we can find that
collects reward at least $k_i/\alpha$; this walk is guaranteed to be no
longer than the optimal walk that collects reward $k_i$ from $V_i$.
We then combine the solutions from each $V_i$ using a dynamic
program very similar to the one described above for times. We omit
details in this version.

\section{The Algorithms}
\label{sec:algos}

In this section, we use the techniques described above to 
develop algorithms which achieve approximation ratios depending on the
lengths of the time-windows. We first consider instances where all
time-windows have integral end-points, and then instances for which
the ratio $L = \frac{\lmax}{\lmin}$ is bounded. Finally, we combine
these ideas to obtain a $\max\{O(\log n), O(\log L)\}$-approximation
for all instances of \tw.

\subsection{An $O(\log L_{max})$-approximation}
\label{subsec:integer-endpoints}

We now focus on instances of \tw~in which, for all vertices $v$,
$R(v)$ and $D(v)$ are integers. Our algorithm requires the following
simple lemma:

\begin{lemma}\label{lemma:interval-partition}
Any interval of length $M>1$ with integral endpoints can be partitioned
into at most $2 \log M$ disjoint sub-intervals, such that the length of any
sub-interval is a power of 2, and any sub-interval of length $2^i$ begins
at a multiple of $2^i$. Further, there are at most 2 sub-intervals of each
length. \end{lemma}
\begin{proof}
Use induction on the length of the interval. The lemma is clearly
true for intervals of length 2 or 3. Otherwise, use at most 2 sub-intervals
of length 1 at the beginning and end of the given interval, so that the
\emph{residual} interval (after the sub-intervals of size 1 are deleted)
begins and ends at an even integer. To cover the residual interval, divide
all integers in the (residual) problem by 2, and apply the induction
hypothesis; we use at most $2 + (2 \log {M/2}) \le 2 \log M$
sub-intervals in total. It is easy to see that we use at most 2 sub-intervals
of each length; intervals of length $2^i$ are used at the $(i+1)$th level
of recursion.
\end{proof}

For ease of notation, we let $\ell$ denote $\log \lmax$ for the rest of
this sub-section. Given an instance of \tw, for each vertex $v$ with
interval $I(v)$, we use lemma~\ref{lemma:interval-partition} to
partition $I(v)$ into at most $2\ell$ sub-intervals. We label the
sub-intervals of $I(v)$ as follows: For each $1 \le i \le \ell$, the
first sub-interval of length $2^i$ is labeled $I^1_i(v)$ and the second
sub-interval $I^2_i(v)$. (Note that there may be no sub-intervals of
length $2^i$.)

We now construct a set of at most $2\ell$ restricted versions of the
given instance. We call these restricted versions $B^1_1, B^1_2,
\ldots B^1_\ell$ and $B^2_1, B^2_2, \ldots B^2_\ell$, such that the
interval for vertex $v$ in $B^b_i$ is $I^b_i(v)$. If $I^b_i(v)$ was not
an interval used in the partition of $I(v)$, $v$ is not present in
the restricted version. (Equivalently, it has reward 0 or an empty
time-window.)

Consider an arbitrary restricted instance $B^b_i$. All vertices in this
instance of \tw~have intervals of length $2^i$, and
all time-windows begin at an integer that is a multiple of $2^i$. Hence,
any 2 vertices either have time-windows that are identical, or entirely
disjoint. This means that $B^b_i$ is a modular instance, so we can
break it into sub-problems, and use a $(2+\eps)$-approximation
to orienteering in the sub-problems to obtain a $(2+\eps)$-approximation
for the restricted instance.

By proposition~\ref{prop:restricted-versions}, one of the restricted
versions has an optimal solution that collects reward at least
$\frac{\opt}{2 \ell}$. Using a $(2+\eps)$-approximation for this
restricted version gives us a
$(2+\eps)\times 2\ell = O(\log L_{max})$-approximation for \tw
when all interval endpoints are integers.

\subsection{A $O(\log n)$-approximation when $L \le 2$}
\label{subsec:boundedL}

For an instance of \tw~when $L =
\frac{L_{max}}{L_{min}} \le 2$, we begin by scaling all times
so that $L_{min} = 1$ (and so $L_{max} \le 2$). Note that even
if all release times and deadlines were integral prior to scaling,
they may not be integral in the scaled version.

For each vertex $v$, we partition $I(v) = [R(v),D(v)]$ into
3 sub-intervals: $I_1(v) = [R(v), a]$, $I_2(v) = [a,b]$, and
$I_3(v) = [b, D(v)]$, where $a = \floor{R(v)+1}$ (that is,
the next integer strictly greater than the release time) and
$b = \ceil{D(v)-1}$ (the greatest integer strictly less than
the deadline). The figure below illustrates the partitioning
of intervals. Note that $I_2(v)$ may be a point, and in this
case, we ignore such a sub-interval.

\begin{figure}[h]
\begin{center}
\begin{tikzpicture}[xscale=2]

%Axes
\draw (0,0) -- (6,0);
\foreach \x in {2,3,5}
{\draw[dashed] (\x,0) -- (\x,1.75);}
\foreach \x in {0,...,6}
{\draw (\x,-0.05) -- (\x,0.05);  \node at (\x,-0.25) {$\x$};  }

\draw (1.5,0.7) -- (3.6,0.7);
\draw (4.4,0.7) -- (5.7,0.7);

\draw[<->] (1.5,1) -- (2,1); \node at (1.75,1.25) {$I_1(u)$};
\draw[<->] (2,1.1) -- (3,1.1); \node at (2.5,1.35) {$I_2(u)$};
\draw[<->] (3,1) -- (3.6,1); \node at (3.3,1.25) {$I_3(u)$};

\draw[<->] (4.4,1) -- (5,1); \node at (4.7,1.25) {$I_1(v)$};
\draw[<->] (5,1) -- (5.7,1); \node at (5.35,1.25) {$I_3(v)$};

\end{tikzpicture}
\end{center}
\caption{We illustrate the partitioning of 2 intervals
into sub-intervals. Note that on the right, $I_2(v)$ is
empty.}
\end{figure}
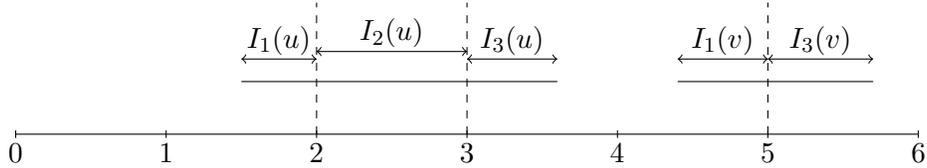

We now construct 3 restricted versions of the given instance ---
$B_1$, $B_2$, and $B_3$ --- such that the interval for any vertex
$v$ in $B_i$ is simply $I_i(v)$. By
proposition~\ref{prop:restricted-versions}, one of these has
an optimal solution that collects at least a third of the
reward collected by an optimal solution to the original
instance. Suppose this is $B_2$. All time-windows have length
exactly 1, and start and end-points are integers. Therefore,
$B_2$ is a modular instance, and we can get a
$(2+\eps)$-approximation to the optimal solution in $B_2$;
this gives a $(6+\eps)$-approximation to the original instance.

Dealing with $B_1$ and $B_3$ is not nearly as easy; they are
not quite modular. Every interval in $B_1$ has length at most
1, and ends at an integer; for $B_3$, intervals have length
at most 1 and start at an integer. We illustrate how to
approximate a solution for $B_3$ within a factor of $O(\log n)$;
the algorithm for $B_1$ is identical except that release times
and deadlines are to be interchanged.

For $B_3$, we can partition the vertex set into $V_1, V_2,
\ldots V_m$, such that all vertices in $V_i$ have the same
(integral) release time, and any vertex in $V_i$ is visited
before any vertex in $V_j$ for $j > i$. Figure 2 shows such a
partition. The deadlines for vertices in $V_i$ may be all
distinct. However, we can solve an instance of \deadline~in
each $V_i$ separately, and paste the solutions together using
dynamic programming. The solution we obtain will collect at
least $\Omega(1/\log n)$ of the reward of an optimal solution
for $B_3$, since there is a $O(\log n)$-approximation for
\deadline~(\cite{timewindow}). Therefore, this gives us a $3
\times O(\log n) = O(\log n)$-approximation to the original
instance.

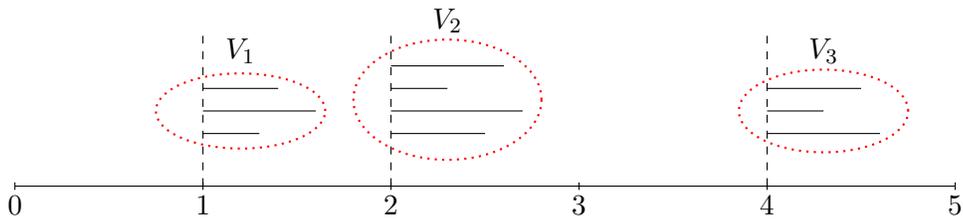
\begin{figure}[h]
\begin{center}
\begin{tikzpicture}[xscale=2.5]

%Axes
\draw (0,0) -- (5,0);
\foreach \x in {1,2,4}
{\draw[dashed] (\x,0) -- (\x,2);}
\foreach \x in {0,...,5}
{\draw (\x,-0.05) -- (\x,0.05);  \node at (\x,-0.25) {$\x$};  }

\draw (1,1.3) -- (1.4,1.3); \draw (1,1) -- (1.6,1); \draw (1,0.7) -- (1.3,0.7);
\draw (2,1.6) -- (2.6,1.6) (2,1.3) -- (2.3,1.3) (2,1) -- (2.7,1) (2,0.7) -- (2.5,0.7);
\draw (4,1.3) -- (4.5,1.3); \draw (4,1) -- (4.3,1); \draw (4,0.7) -- (4.6,0.7);

\draw[dotted,red,thick] (1.2,1) ellipse (0.45cm and 0.5cm); 
\draw[dotted,red,thick] (2.3,1.15) ellipse (0.5cm and 0.8cm); 
\draw[dotted,red,thick] (4.3,1) ellipse (0.45cm and 0.55cm); 

\node at (1.2,1.8) {$V_1$}; \node at (2.3,2.2) {$V_2$}; \node at (4.3,1.8) {$V_3$};

\end{tikzpicture}
\end{center}
\caption{In $B_3$, all time-windows start at an integer and have length
at most 1. Each set of vertices whose windows have a common beginning
corresponds to a sub-problem that is an instance of orienteering with
deadlines.}
\end{figure}

Similarly, we can obtain a $O(\log n)$-approximation for $B_1$
using the $O(\log n)$-approximation algorithm for orienteering
with release times. Therefore, when $L \le 2$, we have a
$O(\log n)$-approximation for \tw.

\subsection{Putting the pieces together}

An arbitrary instance of \tw~may have $L > 2$,
and interval end-points may not be integers. However, we can combine
the algorithms from the two preceding sections to deal with such
instances. We begin by scaling all times such that the shortest
interval has length 1; the longest interval now has length
$L = \frac{\lmax}{\lmin}$, where $\lmax$ and $\lmin$ are the lengths
of the longest and shortest intervals in the original instance.

We now construct 3 restricted versions of the scaled instance: $B_1$,
$B_2$, and $B_3$. For any vertex $v$ with interval $[R(v),D(v)]$ in the
scaled instance, we construct 3 sub-intervals. $I_1(v) = [R(v),a]$,
$I_2(v) = [a,b]$, and $I_3(v) = [b,D(v)]$, where $a = \ceil{R(v) + 1}$
and $b = \floor{D(v) - 1}$. As before, the interval for $v$ in the
instance $B_i$ is $I_i(v)$.

One of the restricted versions collects at least a third of the reward
of the original instance. Suppose this is $B_1$ or $B_3$. All intervals
in $B_1$ and $B_3$ have length between 1 and 2 by our construction.
Therefore, we can use the $O(\log n)$-approximation algorithm from
section~\ref{subsec:boundedL} to collect at least $\Omega(1/\log n)$
of the reward of an optimal solution to the original instance. It now
remains only to consider the case that $B_2$ collects more than a third
of the reward. In $B_2$, the end-points of all time-windows are integral,
and the longest interval has length less than $L$. We can now use the
algorithm of section~\ref{subsec:integer-endpoints} to obtain a
$O(\log L)$-approximation.

Therefore, our combined algorithm is a
$\max\{O(\log n), O(\log L)\}$-approximation for \tw.

\subsection{Towards a better approximation, and arbitrary endpoints}
In the previous sub-section, we obtained an approximation ratio of 
$\max\{(O(\log n), O(\log L)\}$; we would like to improve this ratio
to $O(\log n)$. Unfortunately, it does not seem easy to do this
directly. A natural question, then, would be to obtain a ratio of
$O(\log L)$; this is equivalent to a constant-factor approximation
for the case when $L \le 2$. However, this is is no easier than finding
a $O(\log n)$-approximation for arbitrary instances of \tw, as we show
in the next proposition.

\begin{prop} \label{arbitends}
A constant-factor approximation algorithm for \tw~with $L \le 2$
implies a $O(\log n)$-approximation for arbitrary time-windows.
\end{prop}
\begin{proof}
We show that a constant-factor approximation when $L \le 2$ implies
a constant-factor approximation for \deadline. It follows from an
algorithm of \cite{timewindow} that we can then obtain a
$O(\log n)$-approximation for \tw.

Given an arbitrary instance of \deadline~on graph $G(V,E)$, we add a
new start vertex $s'$ to $G$. Connect $s'$ to $s$ with an edge of length
$D_{max} = \max_v D(v)$. The release time of every vertex is 0, but all
deadlines are increased by $D_{max}$. Observe that all vertices have
time-windows of length between $D_{max}$ and $2 D_{max}$, so $L \le 2$. It
is easy to see that given any walk beginning at $s$ in the original
instance, we can find an equivalent walk beginning at $s'$ in the modified
instance that visits a vertex in its time-window iff the original walk
visited a vertex before its deadline in the given instance, and vice
versa. Therefore, a constant-factor approximation for the modified
instance of \tw~gives a constant-factor approximation for the original
instance of \deadline.
\end{proof}

We \emph{can}, however, obtain a constant-factor approximation for
\tw~when $L \le 2$ if we remove the restriction that
the walk must start and end at $s$ and $t$, the specified
endpoints. The algorithm of Frederickson and Wittman \cite{FW07} for
the case of $L=1$ can be adapted relatively easily to give a
constant-factor approximation for $L \le 2$. For completeness, we
sketch the algorithm here.

We construct 5 restricted versions $B_1, \ldots B_5$, of a given
instance $A$. For every vertex $v$, we create at most 5 sub-intervals
of $I(v)$ by breaking it at every multiple of $0.5$. (For instance
$[3.7,5.6]$ would be broken up into $[3.7,4], [4,4.5], [4.5,5],
[5,5.5], [5.5,5.6]$. Note that some intervals may have fewer than
5 sub-intervals.) The interval for $v$ in $B_1(v)$ is the first
sub-interval, and the interval in $B_5(v)$ is the last sub-interval,
regardless of whether $I(v)$ has 5 sub-intervals. $B_2$, $B_3$, and
$B_4$ each use one of any remaining sub-intervals.

$B_2$, $B_3$, and $B_4$ are modular instances, so if one of them is
the best restricted version of $A$, we can use a $(2+\eps)$-approximation
for orienteering to get reward at least $\frac{\opt(A)}{10 + \eps}$.
Exactly as in subsection~\ref{subsec:boundedL}, $B_1$ and $B_5$ are
not quite modular instances; in $B_1$, all deadlines are half-integral
but release times are arbitrary, and in $B_5$, all release times are
half-integral, but deadlines are arbitrary.

Suppose that $B_1$ is the best restricted version. The key insight is that
if the optimal walk in $B_1$ collects a certain reward starting at $s$ at
time $0$, there is a walk in $B_2$ starting at $s$ at time $0.5$ that
collects \emph{the same} reward. (This is the substance of Theorem 1 of
\cite{FW07}.) Therefore, if $B_1$ is the best restricted version, we
find a $(2+\eps)$-approximation to the best walk in $B_2$ starting at $s$
at time $0.5$; we are guaranteed that this walk collects reward at least
$\frac{\opt(A)}{10+\eps}$. Note that this walk may not reach the destination
vertex $t$ by the time limit, since we start $0.5$ time units late.
Similarly, if $B_5$ is the best restricted version, we can find a walk in
$B_4$ that collects reward $\frac{\opt(A)}{10+\eps}$ while beginning at
$s$ at time $-0.5$. (To avoid negative times, we can begin the walk at $s'$
at time 0, where $s'$ is the first vertex visited by the original walk
after time 0.) This walk is guaranteed to reach $t$ by the time limit, but
does not necessarily begin at $s$.

Therefore, this algorithm is a $O(1)$-approximation when $L \le 2$, or
a $O(\log L)$-approximation for general instances of \tw~with arbitrary
endpoints. We note that one \emph{cannot} use
this with proposition~\ref{arbitends} to get a $O(\log n)$-approximation
for \tw~with arbitrary endpoints. The dynamic program
for modular instances crucially uses the fact that we can specify both
endpoints for the sub-problems.

\section{Conclusions}
\label{sec:conclusion}

We conclude with some open problems:
\begin{itemize}
\item Can we obtain an improved approximation ratio for the general
  time-window problem that does not depend on the lengths of the
  intervals? In particular, is there an $O(\log n)$ (or even an
  $O(\log \opt)$) approximation for undirected graphs?

\item The current algorithms for \tw~use the
  orienteering algorithm in a black-box fashion. For directed graphs
  the current ratio for orienteering is $O(\log^2 \opt)$ and hence the
  ratio for \tw~is worse by additional logarithmic
  factors. Can we avoid using orienteering in a black-box fashion?  We
  note that the quasi-polynomial time algorithm of \cite{ChekuriP05}
  has the same approximation ratio of $O(\log \opt)$ for both the
  basic orienteering and time-window problems in directed graphs.  In
  fact, the algorithm gives the same approximation ratio even for the
  more general problem where each vertex has multiple disjoint time
  windows in which it can be visited.

\item For many applications, each vertex has multiple disjoint time-windows,
  and we receive credit for a vertex if we visit it within
  any of its windows. If each vertex has at most $k$ windows, a
  na\"{i}ve algorithm loses an extra factor of $k$ beyond the ratio
  for \tw, but no better approximation is known.
  Any non-trivial result would be of interest.
\end{itemize}

\medskip
\noindent
\textbf{Acknowledgments:} We thank Pratik Worah for several discussions of
algorithms and hardness results for orienteering with time-windows and other
variants.

\bibliographystyle{plain}

\end{document}